%% file: manuscript.tex
\documentclass[11pt]{article}
\usepackage[margin=1in]{geometry}
\usepackage{amsmath,amssymb,amsthm,graphicx,booktabs,tabularx}
\usepackage[T1]{fontenc}
\usepackage{lmodern}
\usepackage[colorlinks=true,allcolors=blue]{hyperref}
\usepackage{microtype}
\newtheorem{proposition}{Proposition}
\newtheorem{corollary}{Corollary}
\newcommand{\E}{\mathbb E}

\newcommand{\CI}{G_{\mathrm{CI}}}
\newcommand{\Iv}{\mathcal I_{\mathrm{vote}}}
\newcommand{\Rp}{R_{\mathrm{probe}}}
\title{When Does Communication Help? Beyond Spectral Descriptions of Collective Intelligence}
\author{Xuening Wu\\[0.4em]\small AI Engineering, Pfizer, Shanghai, 200041\\[0.2em]\small\href{mailto:xuening.wu@pfizer.com}{\texttt{xuening.wu@pfizer.com}}}
\date{}
\begin{document}
\maketitle
\begin{abstract}
Communication can bring agents into agreement while making their decisions worse. We identify two limits of aggregate descriptions of communication gain in distributed inference. First, stable linear systems with fixed evidence, network and readout can have interaction and finite-time state operators with identical eigenvalue and singular-value spectra, yet produce gains of opposite sign. Changing only message orientation raises accuracy from 72.6\% to 91.2\% or lowers it to 65.9\%. A standard task-projected local-response approximation retains the directional information missing from spectral summaries. Using labeled calibration data separate from the test set, it predicts multi-round gains in small trained nonlinear agents with a root-mean-square error of 0.45 percentage points on two synthetic tasks; tests with natural edge changes and handwritten digits extend the evaluation. Second, under community-shared bias, higher mean individual accuracy can coexist with harm to unaffected communities or lower global-vote accuracy. At fixed communication rounds, calibration constraints reduce observed community harm while retaining much of the mean benefit, but do not guarantee protection. Full direct calibration performs similarly. The results connect spectral insufficiency, task-aware prediction and the distribution of communication benefits, while leaving broad transfer and practical superiority open.
\end{abstract}

\section{Introduction}
Communication gives agents access to one another's evidence, but it also creates dependence between their errors. A group can therefore become more coordinated without becoming more accurate. Human experiments document both convergence without improved accuracy \cite{lorenz2011} and improvements in judgment under suitable network influence structures \cite{becker2017}. For learned agents, the corresponding question is practical: can we predict whether a particular communication process will improve performance on a particular task?

Three mechanisms organize this question. Local amplification governs how strongly messages change an agent's state. Network propagation governs how far those changes travel within a finite communication budget. Private-evidence diversity determines what complementary information is available. Yet a scalar measure of each mechanism may discard information essential to prediction: spectral radius omits direction, graph mixing omits transformations between feature channels, and a statistic of private evidence alone omits the path from that evidence to the task readout.

Our contribution is a controlled test of what a predictive description of communication gain must retain. First, an exact counterexample shows that complete eigenvalue and singular-value spectra can leave opposite gains indistinguishable. Second, we test whether a standard local-response approximation that retains the task projection predicts multi-round gains in trained nonlinear agents. Third, community-resolved experiments show that mean benefit can conceal local harm and test calibration constraints intended to limit that harm. The resulting account connects the information omitted by spectral summaries to prediction and to the distribution of communication benefits.

We measure mean individual task accuracy. This quantity can improve when communication distributes evidence already present in the group. Accordingly, ``collective intelligence gain'' denotes the measured effect of communication; it does not imply performance beyond centralized aggregation or the emergence of a new general capability.

\section{Related Work}

Agreement and accuracy are distinct in social-learning experiments: social influence can reduce diversity without improving accuracy \cite{lorenz2011}, while suitable networks can improve judgments \cite{becker2017}. Distributed detection theory relates network averaging speed to asymptotic error rates under fresh observations \cite{bajovic2011}. Our main experiments instead hold private evidence fixed and measure finite-round changes in mean individual accuracy. The qualitative importance of amplification, propagation and evidence diversity is established; the question here is what their summaries omit.

Transient neural coding and aligned or oblique recurrent dynamics establish the functional role of input, propagation and readout geometry \cite{bondanelli2020,schuessler2023}. Control-theoretic interpretation uses related input--output tools \cite{moon2025}, while graph neural network (GNN) studies analyze information bottlenecks \cite{black2023} and representation convergence \cite{epping2024}. Our contribution is a controlled construction with matched complete spectra and opposite communication gains, followed by evaluation of a standard local-response approximation for that endpoint, rather than a new principle of linearization or readout geometry.

Agent-system studies examine architecture-dependent returns to coordination \cite{kim2026}, and debate studies emphasize voting baselines and correction-biased updates \cite{choi2025,liu2026}. CAGE-CAL compares agent dependency graphs with and without communication to calibrate language-model panel confidence and select topologies \cite{huang2026cage}; BOUNDARY\_SYNC measures representational homogenization \cite{liu2026boundary}. False consensus and calibration-guided selection therefore overlap with prior work. We study communication-induced accuracy changes using private observations, white-box neural operators and target-condition labels, not a label-free law or a general advantage over aggregation.

Worst-group learning addresses limitations of average performance \cite{sagawa2019}, and Learn then Test provides finite-sample risk control through calibration and multiple testing \cite{angelopoulos2021}. Our groups are communicating node communities, and the tested constraints concern their accuracy changes relative to within-only communication. We evaluate standard constraints and approximate margins empirically; we do not derive a new risk-control theorem or a certified guarantee.

\section{Method}
\subsection{Problem formulation and candidate coordinates}
Each episode contains a target $Y$ and $N$ agents with private observations $o_i$. A shared encoder maps each observation to a private state. Agents exchange messages over a connected undirected graph with adjacency matrix $A$ and row-stochastic random-walk operator $P=D^{-1}A$, where $D$ is the degree matrix. For binary decisions $\widehat Y_i^{(T)}$ after $T$ rounds, define
\begin{equation}
 Q_T=\E\left[\frac{1}{N}\sum_{i=1}^N \mathbf 1\{\widehat Y_i^{(T)}=Y\}\right],\qquad
 \CI=Q_T-Q_0.
 \label{eq:gain}
\end{equation}
The primary baseline sets message coupling to zero while retaining the same trained parameters. This comparison isolates the inference-time effect of communication. It does not establish an advantage over a separately optimized isolated model. The data also include separately trained isolated models, global majority votes over private decisions, and averages of private prediction probabilities.

In row-state notation, the learned dynamics are
\begin{equation}
 h^0=\phi(b),\quad h^{t+1}=\phi(b+gPh^tW),\quad \ell_i^t=h_i^t v+c,
 \label{eq:neural}
\end{equation}
where $b$ is the fixed private anchor, $g$ controls message coupling, $W$ transforms message channels, and $v$ and $c$ define the task readout. For labels in $\{0,1\}$, the decision is $\widehat Y_i^t=\mathbf 1\{\ell_i^t>0\}$ and the predicted probability is the logistic transform of $\ell_i^t$. Execution uses local messages, while training uses a shared supervised objective. The update is a neural message-passing rule without Ising spin gates or an Ising energy.

We evaluate three prespecified operational coordinates. Let $\operatorname{vec}_{\mathrm{ag}}$ stack the agent rows into a column vector. With $\overline D=\operatorname{diag}(\operatorname{vec}_{\mathrm{ag}}\E_{\mathrm{cal}}[\phi'(b)])$, where the expectation averages calibration episodes, define
\begin{align}
 \Lambda&=g\rho\!\left(\overline D(P\otimes W^\top)\right),\label{eq:lambda}\\
 \mu_T&=1-\operatorname{SLEM}(P)^T,\label{eq:mu}\\
 \Iv&=Q_{\mathrm{cal}}(\text{majority of private decisions})-Q_{0,\mathrm{cal}}.\label{eq:info}
\end{align}
Here, $\rho$ denotes spectral radius and $\operatorname{SLEM}$ the second-largest eigenvalue modulus; tied votes receive half credit. The coordinates have distinct limits. $\Lambda$ measures reference susceptibility at private preactivations, rather than the Jacobian spectrum along the interacting trajectory. $\mu_T$ summarizes spectral exposure, rather than total-variation mixing or the full coupled dynamics. $\Iv$ measures a signed, label-dependent pooling gain, rather than mutual information. Its predictive value may therefore partly reflect direct measurement of an alternative aggregation strategy.

The definitions remain fixed throughout the transfer and response studies. Predictive surfaces $F$ are fitted using gains from their training folds only. The response approximation introduced below uses additional information; it is a separate predictor, not a revised definition of the original coordinates.

\subsection{A matched-spectrum counterexample}\label{sec:counterexample}
Consider two-channel observations, now in column-state notation,
\begin{equation}
 x_{i1}=\alpha Y+\epsilon_i,\qquad
 x_{i2}=\sigma\big(\sqrt r\,u+\sqrt{1-r}\,\eta_i\big),
 \label{eq:data}
\end{equation}
where $Y\in\{-1,1\}$ has equal priors, $\sigma>0$ and $r\in[0,1]$. The variables $\epsilon_i$, $u$ and $\eta_i$ are independent standard normals, all independent of $Y$. The second channel is therefore uninformative about the target. The private decision is $\operatorname{sign}(x_{i1})$, with accuracy $\Phi(\alpha)$, where $\Phi$ is the standard normal cumulative distribution function.

Let $h_i^0=x_i$ and
\begin{equation}
 h_i^{t+1}=x_i+g\sum_jP_{ij}W_\theta h_j^t,\qquad
 W_\theta=R_\theta\begin{pmatrix}a&0\\0&b\end{pmatrix}R_\theta^\top,
 \label{eq:linear}
\end{equation}
with $a=0.8$, $b=0.2$ and $R_\theta=\left(\begin{smallmatrix}\cos\theta&-\sin\theta\\\sin\theta&\cos\theta\end{smallmatrix}\right)$. The scalars $a,b$ here denote channel eigenvalues; the scalar $b$ is unrelated to the private anchor in Eq.~\eqref{eq:neural}. Decisions always read the first state coordinate. For $\theta\in[0,\pi/2]$, the channel operator is positive definite and entrywise nonnegative. With $0<g\leq1$, the recurrence is stable.

\begin{proposition}[Spectral matching does not determine gain]\label{prop:spectral}
Let $\alpha>0$, $0<r\leq1$ and $N\geq2$. Fix a connected undirected graph, the first-coordinate readout specified above, $0<g\leq1$ and an integer $T\geq1$. For each choice of noise parameters, hold the observations fixed across orientations. Changing $\theta$ in Eq.~\eqref{eq:linear} preserves all eigenvalues and singular values of both the one-step Jacobian and the finite-time state operator. Nevertheless, for sufficiently strong shared nuisance in Eq.~\eqref{eq:data}, the gains at $\theta=0$ and $\theta=\pi/4$ have opposite signs.
\end{proposition}
The proof is given in Appendix~\ref{app:spectral-proof}. Its key observation is that $J_\theta=g(P\otimes W_\theta)$ and $S_T(\theta)=\sum_{k=0}^{T}J_\theta^k$ change by orthogonal similarity as the message orientation changes. Their eigenvalues and singular values are therefore preserved. With the input and readout fixed, however, the rotation changes how much shared nuisance reaches the decision. At zero degrees, the readout pools informative observations; at 45 degrees, sufficiently strong shared nuisance makes accuracy worse than the private baseline.

For the linear update, the activation derivative is the identity, so the reference susceptibility reduces to $\Lambda=g\rho(P\otimes W_\theta)=0.8g$. The other candidate coordinates are also unchanged: the graph and $T$ determine $\mu_T$, and the private decisions determine $\Iv$. In fact, the entire input distribution and its information about $Y$ remain fixed. Rotating only the interaction operator changes its orientation relative to the task. Rotating observations and readout together with the operator would instead be a passive change of basis and leave decisions unchanged; a numerical check verifies this distinction.

\begin{corollary}[Scalar prediction lower bound]
For two such systems with gain difference $\Delta$, any deterministic predictor based only on the three matched coordinates has maximum absolute error at least $|\Delta|/2$. The smallest attainable equal-weight two-condition RMSE among such predictors is $|\Delta|/2$.
\end{corollary}
Both systems receive the same prediction, and the midpoint of their gains minimizes squared error. The construction therefore establishes insufficiency of the specified summaries within this model class. The underlying orthogonal-similarity argument and Gaussian classification formulas are standard.

\subsection{Task-aware finite-time response}
For each calibration episode, let $H(h)$ denote the nonlinear update in Eq.~\eqref{eq:neural}. One message step from the private state gives the displacement and local derivative
\begin{equation}
 d=H(h^0)-h^0,\qquad
 J[\delta]=\phi'(b+gPh^0W)\odot(gP\delta W).
\end{equation}
The derivative $J=DH(h^0)$ describes the update's sensitivity to its input state and differs from the reference operator in Eq.~\eqref{eq:lambda}. Holding this derivative fixed gives the affine approximation
\begin{equation}
 \widetilde\delta^0=0,\qquad
 \widetilde\delta^{t+1}=d+J[\widetilde\delta^t],\qquad
 \widetilde\ell_i^T=\ell_i^0+\widetilde\delta_i^T v.
 \label{eq:response}
\end{equation}
Equivalently, $\widetilde\delta^T=\sum_{k=0}^{T-1}J^k[d]$. We define $\Rp$ as the empirical accuracy gain of these surrogate logits on labeled calibration data, independent of the final test set. No regression mapping response to gain is fitted, and the formula is unchanged across tasks and activation functions.

The approximation retains the initial displacement, its subsequent propagation and its projection onto the task readout. It is exact after one round and generally approximate thereafter. Because accuracy also depends on labels and decision margins, alignment alone does not guarantee a benefit. The predictor requires white-box parameters, local derivatives, calibration labels and one actual message round per candidate coupling and operator. It is therefore a local surrogate simulator with a larger information budget than the three scalar coordinates.

We use three ablations and a direct-calibration reference. The one-step baseline reuses the measured first-round calibration gain for every $T$. The readout ablation projects the propagated displacement onto an orthogonal hidden direction while preserving private logits. The zero-message approximation evaluates slopes at the private preactivation. Finally, full direct calibration executes every nonlinear candidate for all rounds on the same labeled samples, providing a reference for simply measuring candidate performance.

\subsection{Community-resolved communication effects}\label{sec:community-method}
For a partition into four equally sized communities $C_k$, let $P_{\rm in}$ average the preceding and following members within each community, and $P_{\rm out}$ average the same-ranked members of its two neighboring communities. Repeated neighbors are retained as message slots when a community has two members. Define
\begin{equation}
 P_\beta=(1-\beta)P_{\rm in}+\beta P_{\rm out},\qquad
 \Delta_k(\beta,T)=Q_{k,T}(P_\beta)-Q_{k,T}(P_{\rm in}),
 \label{eq:community}
\end{equation}
where $Q_{k,T}$ is mean individual accuracy within $C_k$. We use $\beta$ for the inter-community weight to distinguish it from the signal parameter $\alpha$ in Eq.~\eqref{eq:data}; result files call this weight \texttt{alpha}. At $\beta=0$ the communities are disconnected, extending the connected-graph setting above. The coordinate $\mu_T$ is not used to define the community selector. Unlike $\CI$, $\Delta_k$ compares inter-community communication with within-community communication at the same $T$. A separate endpoint computes a majority vote over all agents, with tied votes receiving half credit.

The follow-up models use either Eq.~\eqref{eq:neural} with Softplus, or the residual update
\begin{equation}
 h^0=\tanh(b),\qquad h^{t+1}=(1-g)h^t+g\tanh(b+Ph^tW).
 \label{eq:residual}
\end{equation}
For both, Eq.~\eqref{eq:response} uses the derivative of the implemented update. Changing activation and state retention together tests a second update form; it does not isolate either factor's causal contribution.

With target-condition calibration data, point-constrained selection chooses
\begin{equation}
 \widehat\beta=\arg\max_{\beta\in\{0,0.25,0.5,1\}}\widehat Q_T(P_\beta)
 \quad\text{subject to}\quad \min_k\widehat\Delta_k(\beta,T)\geq-\varepsilon.
 \label{eq:constraint}
\end{equation}
Predictions come from either local response or full nonlinear calibration on the same samples. The selector receives community membership and calibration outcomes, but neither test labels nor the identity of the corrupted communities. Within-only communication is always feasible; no-communication is not a candidate in this experiment. Ties select the first candidate in increasing $\beta$ order. A guarded variant replaces each estimated change by $\widehat\Delta_k-z\widehat{\rm SE}_k$, with $z=\Phi^{-1}(1-0.05/12)$ for the three nonbaseline candidates and four communities. Standard errors use paired per-episode community accuracies from the respective surrogate or nonlinear predictions. Normal approximation and response error preclude a finite-sample protection guarantee.

\section{Experiments}
\subsection{Experimental setup}
\subsubsection{Sequence and separation of evidence}
The studies were developed sequentially. A preliminary experiment with Gaussian private evidence and a binary task derived from $8\times8$ handwritten digits \cite{digits1998} motivated the linear counterexample. The counterexample then motivated the response experiment, which used new task generators and training seeds. Protocols, configurations and implementation hashes were recorded locally before each production run. These records document the sequence of development; they are not external preregistration, and earlier outcomes informed later study designs.

The preliminary study evaluated 5,184 conditions with a shared eight-channel Softplus network. Agents received disjoint pixel subsets of digits, and training, calibration and test sets used disjoint source-image identities. The conditions covered population sizes 8 and 16, three graph families, redundancy, misleading subpopulations and common-mode corruption. The reproducibility archive contains the full settings and supplementary predictions. We report task-transfer failures separately from pooled performance.

The linear study comprised 441 conditions. Its preselected primary condition used $N=12$, a degree-four ring, $g=1$, $T=4$, $\alpha=0.6$, $\sigma=4$ and $r=0.5$. This was a deliberate stress test chosen from the sufficient inequality in Eq.~\eqref{eq:sufficient}. The primary sweep used 131,072 paired episodes; replication blocks used 16,384 episodes each. Comparisons covered seven orientations and three graph families. Analytic Gaussian accuracies provided a separate check of the simulation implementation.

\subsubsection{Learned nonlinear response study}
The response study uses two synthetic tasks requiring nonlinear processing of private features, with labels $Y\in\{0,1\}$. In the XOR task, agent $i$ observes $(s_i+\epsilon_i,(2Y-1)s_i+\eta_i,n_{i1},n_{i2})$, where $s_i$ are independent equiprobable signs and the noise terms are independent standard Gaussian variables. In the signal task, each agent observes 16 samples of a sinusoid with two or three cycles per window, determined by $Y$, an independent random phase, and Gaussian noise of standard deviation 1.5.

We evaluate independent evidence and two controlled distribution shifts. Common nuisance adds a shared two-dimensional Gaussian vector, with independent components of standard deviation two, to the last two XOR channels, or a shared random-phase sinusoid of amplitude two to the signal task, with frequency independent of $Y$. The misleading condition flips the source label for the first 75\% of agent identities. Models train only on independent evidence. These shifts probe specific failure mechanisms rather than approximate deployment distributions.

The shared encoder has a 16-unit ReLU layer followed by an eight-channel affine layer. State activation $\phi$ is either Softplus or ELU. All encoder, message and readout parameters are trained using Adam for 800 minibatches of 64 episodes, with a learning rate of 0.003 and the message matrix spectral norm clipped to 0.75. Training uses $N=8$ Erd\H{o}s--R\'enyi graphs, $g\in\{0,0.35,0.7\}$ and $T\in\{1,2,3\}$. Three seeds give 12 learned models across task and activation combinations, plus matched independently trained no-communication models. These are two activations within a common update template, not two substantially different communication architectures.

Evaluation crosses $N\in\{8,16\}$, degree-four rings or connected Erd\H{o}s--R\'enyi graphs sampled with edge probability $4/(N-1)$ and conditioned on connectivity, three evidence scenarios, $g\in\{0,0.35,0.7,1.1\}$, $T\in\{1,3,6\}$ and five operator interventions. The calibration and test sets contain 2,048 and 4,096 independently generated episodes per block, respectively. A separate 1,024-episode sample defines the orientation basis. Evidence draws are shared across graph and activation comparisons where specified, preserving paired comparisons but limiting statistical independence.

The five interventions are the trained operator, rotations by 45 or 90 degrees, and scaled versions of the two rotations that match the trained operator's reference $\Lambda$. Rotations act in the span of the normalized task readout and a leading private-state variance direction orthogonal to it. This second direction is not guaranteed to be independent of the label. Raw rotations preserve the spectrum of $W$ but not the nonlinear Jacobian. Scaled rotations match only the reference radius; they can change message norms. The stronger invariance claim of the linear construction does not apply to them.

All 8,640 conditions are retained in the archive. The primary analysis uses the 4,320 conditions with $g>0$ and $T>1$, excluding zero-coupling cases and the exact one-step case. Fitted controls use histogram gradient boosting with fixed parameters and no hyperparameter search. For simultaneous task/activation holdout, training excludes both the held-out task and the held-out activation. The response predictor requires no fitted surface but does use calibration observations from the target condition. The transfer claim concerns the formula with target-condition calibration, not label-free prediction.

\subsubsection{Follow-up studies and data reuse}
The community studies use 12 models trained on XOR or binary handwritten digits, with anchored or residual updates and three new training seeds. At $N\in\{8,16\}$ and $g=0.7$, four communities exchange within- and between-community messages. One or two communities receive a half- or full-strength bias, alongside a clean control. XOR bias suppresses or reverses the second observation channel; image bias shifts observed pixels in the direction of the difference between the training positive- and negative-class means, without consulting test labels. These are directed stress tests. The bias study evaluates 1,440 configurations at $T\in\{1,3,6\}$; the constraint study uses three new partitions and fixes $T=3$ or $6$ separately, yielding 720 blocks and 2,880 candidates.

Both studies use 128 calibration episodes and 1,024 XOR test episodes or 450 digits test images. The primary constraint tolerance is $\varepsilon=0.01$. Four scheduled message slots per node per round are charged even at zero weight, so slots are matched only within each $(N,T)$; active edges differ. Later image studies reuse the same test identities, and XOR keeps evidence per agent fixed whereas digits keeps total pixels fixed. Natural edge interventions provide an additional response check. Appendix~\ref{app:followup} gives full training, corruption and sampling protocols, and Appendix~\ref{app:dynamic} describes a separate recovery stress test.

\subsubsection{Metrics and uncertainty}
We report the coefficient of determination ($R^2$), root-mean-square error (RMSE) and balanced accuracy for the sign of gain. Balanced accuracy averages recall for positive and negative gains. Sign evaluation excludes gains with absolute magnitude at most 0.01 in the response study; the preliminary study used its separately specified threshold of 0.005. Paired accuracy intervals use episodes as sampling units; sign balanced accuracy is left undefined when only one sign is present beyond the specified threshold. The parameter-grid rows are correlated observations of 12 trained response models, not independent replications. Descriptive bootstrap summaries resample model blocks; shared evidence across activations limits their interpretation beyond the tested tasks.

The response protocol specified two sets of success criteria. Predictive criteria required positive $R^2$ in all four task/activation groups, sign balanced accuracy of at least 0.75, and RMSE at least 20\% below the better of the three-coordinate and raw-control surfaces under simultaneous holdout. Mechanism criteria required RMSE at least 10\% below both the one-step and wrong-readout baselines. These criteria assess performance within the study and do not establish broad generality or practical superiority.

\subsection{Spectral insufficiency}
\subsubsection{Scalar transfer is condition-dependent}
Strong prediction within grouped conditions did not ensure task transfer. The preliminary three-coordinate surface achieves grouped-condition $R^2=0.9545$, but training on Gaussian evidence and testing on digits gives $R^2=-0.2702$; the reverse direction gives $R^2=0.5469$. Replacing $\Lambda$ with raw coupling slightly lowers pooled grouped-condition RMSE, from 1.418 to 1.395 percentage points. Prediction therefore does not transfer reliably in both directions, and the chosen amplification scalar shows no unique predictive advantage in this comparison. Because fitting error or differences in condition coverage could also explain this failure, we next turn to the exact counterexample.

\subsubsection{Identical spectra coexist with opposite gains}
In the preselected linear contrast, the same private decisions have accuracy 72.607\%. Interaction yields 91.180\% at zero degrees and 65.882\% at 45 degrees. Gains are $+18.572$ and $-6.725$ percentage points, with a paired difference of 25.297 percentage points (nominal 95\% confidence interval 25.080--25.514). All three coordinates are identical: $(\Lambda,\mu_T,\Iv)=(0.8,0.782372,0.221737)$, as are the complete one-step and finite-time eigenvalue and singular-value spectra. The third value is the calibration estimate defined in Eq.~\eqref{eq:info}; its analytic population counterpart is 0.221697. Both are unchanged across orientations.

The exact Gaussian gain difference is 0.253373. By the corollary, any deterministic function of the three coordinates has equal-weight RMSE of at least 12.669 percentage points on this pair, irrespective of fitting capacity. The preselected pair also has opposite gains in all 12 predeclared strong shared-nuisance replication combinations. These results establish that such failures are possible and that the summaries are insufficient; they do not estimate how often the failures occur naturally.

\begin{figure}[tb]
\centering\includegraphics[width=\linewidth]{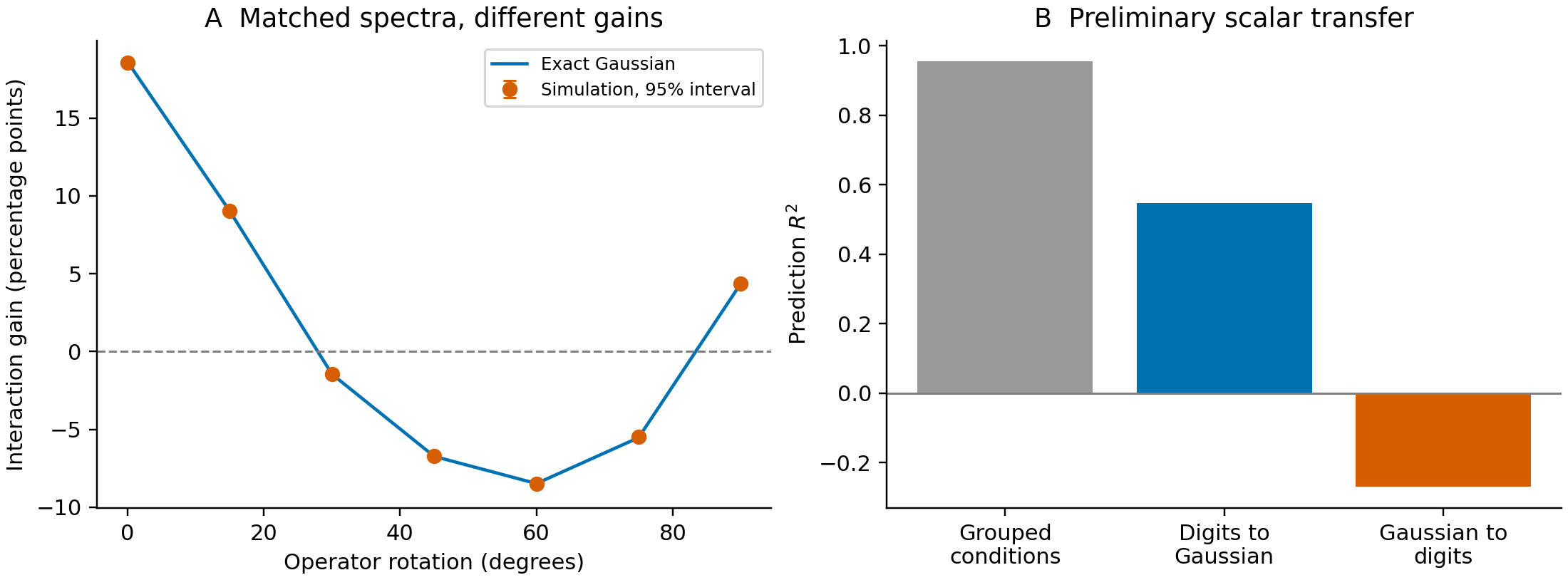}
\caption{Scalar insufficiency and preliminary task transfer. Left: primary linear gains across all predeclared orientations; points are Monte Carlo estimates and the line is the exact Gaussian calculation. The dashed line is zero gain. Graph, evidence, readout and full operator spectra are fixed. Right: preliminary three-coordinate prediction across two task-transfer directions. Grouped-condition prediction is shown separately because it is a different generalization test.}
\label{fig:linear}
\end{figure}

\subsection{Task-aware response predicts multi-round gains}
On the 4,320 primary neural conditions, $\Rp$ achieves $R^2=0.9988$ and RMSE 0.454 percentage points (Table~\ref{tab:prediction}; Fig.~\ref{fig:response}). The four task/activation $R^2$ values range from 0.9985 to 0.9989. Both sets of prespecified criteria are met. The three-coordinate fitted surface has $R^2=-0.1044$ under simultaneous task/activation holdout, while raw controls reach 0.8389. These comparisons demonstrate the value of additional response information in this setting; they do not compare equal-information predictors.

\begin{table}[tb]
\centering\small
\input{tables/prediction.tex}
\caption{Primary response study. RMSE is in percentage points; BA is sign balanced accuracy in percent. Fitted controls use simultaneous task/activation holdout. Unfitted predictors use the target-condition calibration set without fitting a response-to-gain surface. The predictors have different information access.}
\label{tab:prediction}
\end{table}

\begin{figure}[tb]
\centering\includegraphics[width=0.94\linewidth]{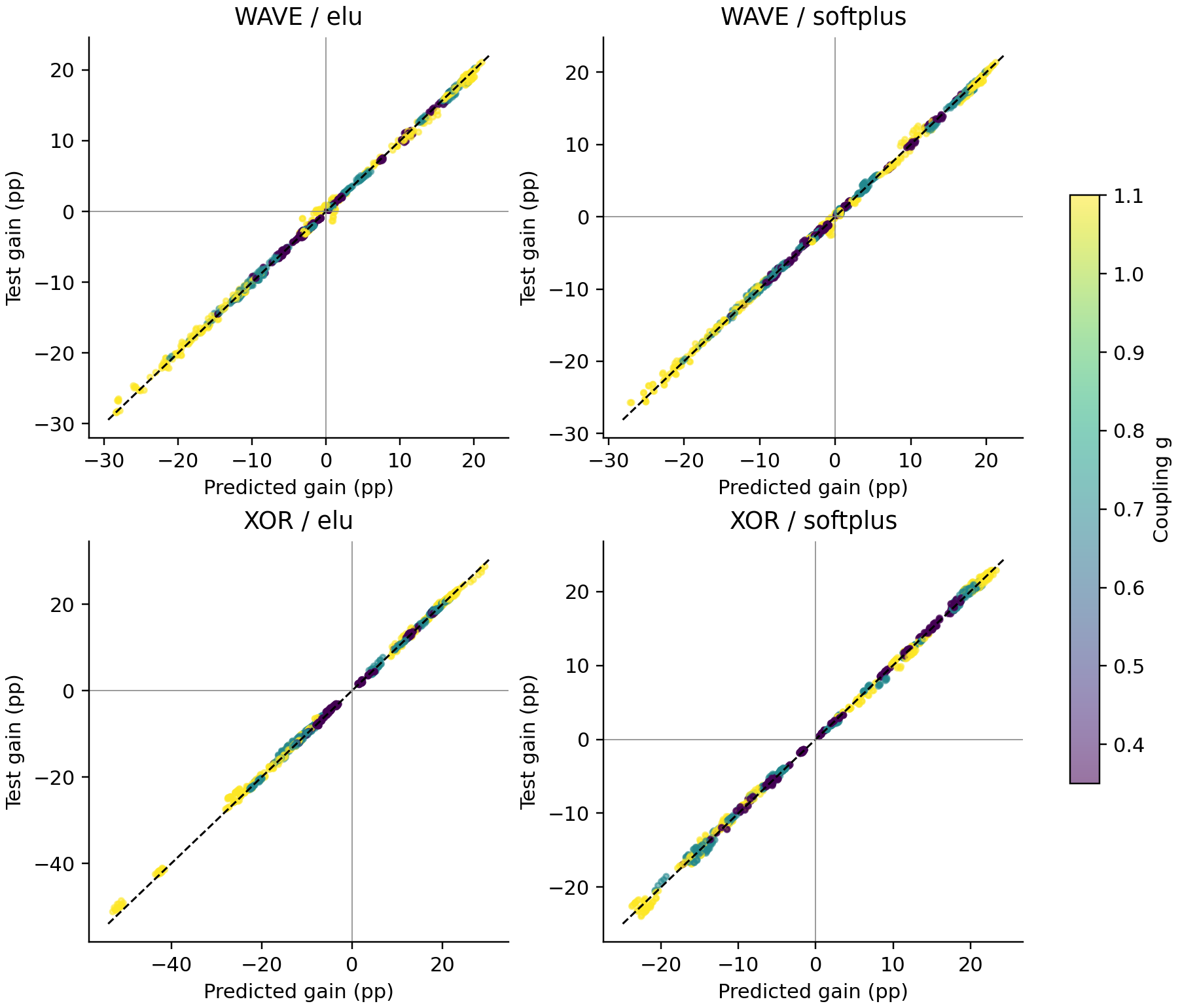}
\caption{Calibration response predictions versus independent-test gains, separately by task and activation. Each panel contains 1,080 primary conditions. Color denotes coupling and the diagonal denotes perfect prediction. Dense, related parameter sweeps should not be interpreted as independent task replications.}
\label{fig:response}
\end{figure}

One-step gain alone predicts signs well (99.61\% balanced accuracy), but has RMSE 4.499 percentage points. Propagating the response mainly improves magnitude prediction. The orthogonal-readout ablation has RMSE 16.798 points. The zero-message approximation also performs well (0.568 points), indicating that these systems broadly admit accurate local approximations rather than uniquely favoring one expansion point. Full direct calibration is slightly more accurate than $\Rp$ at 0.426 points.

Response RMSE increases from 0.278 percentage points at $g=0.35$ to 0.623 at $g=1.1$. Exploratory analyses conducted after evaluation also show high predictive accuracy for unrotated trained operators ($R^2=0.9994$) and within each evidence scenario. These analyses did not determine the predictor or its success criteria.

\subsubsection{Agreement and functional gain remain distinct}
Among 1,728 primary intervention pairs matched on reference radius, 811 have gains of opposite sign, with both nominal per-condition confidence intervals excluding zero. These dependent contrasts were not corrected for multiple comparisons. They indicate sensitivity to operator orientation, while providing weaker control than the exact spectral matching in the linear construction.

In 475 primary conditions, decision agreement exceeds 90\% while gain falls below $-1$ percentage point. This is a descriptive count within the tested grid, not a deployment frequency estimate. Agreement is the mean fraction of distinct agent pairs making the same binary decision, averaged over episodes; gain measures improvement over private decisions. Their separation is consistent with both the counterexample and prior social-learning evidence.

\paragraph{Natural-intervention check.}
On digits with edge deletion and coupling changes, response RMSE is 1.022 percentage points for anchored and 1.395 for residual updates, using 128 calibration labels and excluding $T=1$ from prediction evaluation. Selection includes all candidate round counts, including $T=1$, and abstention. Selected gains over private decisions are 7.272 and 7.306 points, versus 7.343 and 7.308 for direct calibration. Global private-probability averaging remains stronger than the selected mean individual readout, and the grid contains no gain below $-1$ point, limiting evaluation of harmful-communication detection. Observed timing advantages occur only for anchored updates. Appendix~\ref{app:followup} reports the complete comparison.
\subsection{Community benefits and harms}\label{sec:community-results}
\subsubsection{Mean benefit can conceal community harm}
Among 576 biased, multi-round comparisons of inter-community versus within-only communication, 189 reduce unaffected-community accuracy by more than one percentage point. In 74 comparisons, mean individual accuracy improves by more than one point while unaffected-community accuracy falls by more than one point. These are dependent grid counts, not estimates of a deployment event rate.

One illustrative slice uses one community exposed to full-strength bias, $\beta=0.5$, $T=6$, $N=16$ and anchored XOR agents. Averaged across three training seeds, affected-community accuracy rises by 19.98 points and unaffected-community accuracy falls by 2.29 points. Mean individual accuracy rises by 3.28 points, but global-majority accuracy falls by 1.25 points. This slice was selected after analysis for illustration. Figure~\ref{fig:community} instead displays all biased multi-round contrasts. The result shows why both the beneficiary and the readout must be specified when reporting communication gains.

\begin{figure}[tb]
\centering\includegraphics[width=\linewidth]{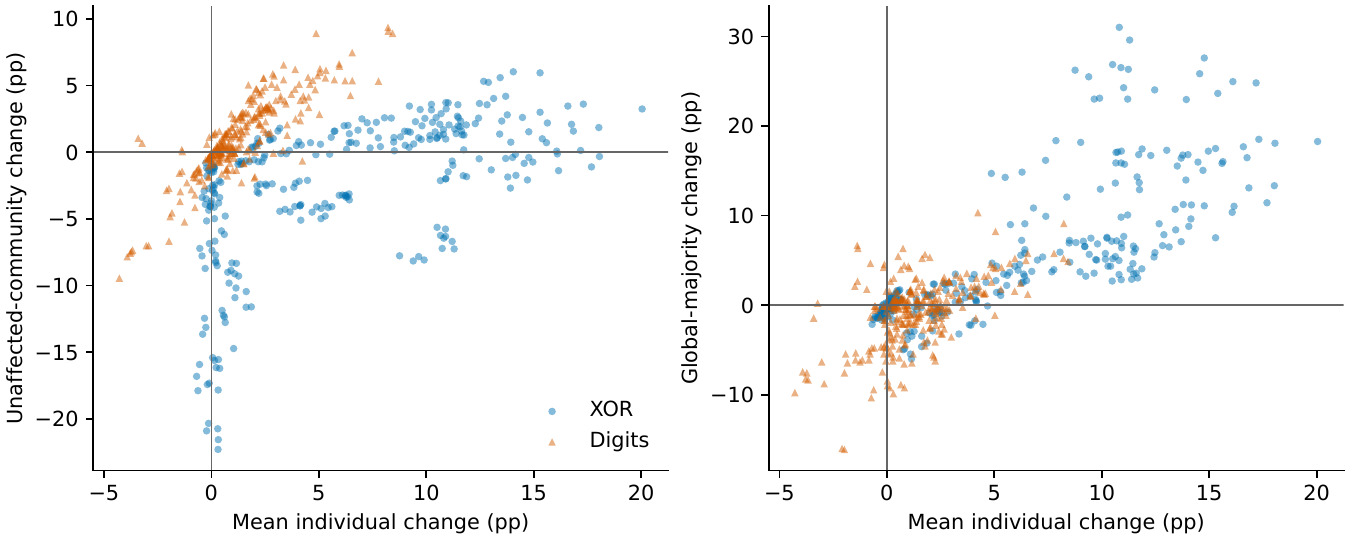}
\caption{Community and readout effects in all 576 biased comparisons with $T>1$. Each point compares a nonzero inter-community weight with within-only communication at the same $T$. Left: positive mean individual effects can coexist with negative effects on unaffected communities. Right: positive mean individual effects need not improve global-majority accuracy. Related conditions share models and evidence; points are not independent replications.}
\label{fig:community}
\end{figure}

On all multi-round cross-community contrasts, including clean controls, response RMSE for unaffected-community effects is 1.59 and 1.72 points for anchored and residual digits agents, and 1.01 and 1.21 points for XOR. Relative to no communication, nine digits configurations lose more than one point overall, with a minimum gain of $-2.67$ points; these all have $N=8$, two affected communities and $\beta=1$. This last comparison uses $\CI$, not the within-only reference in Eq.~\eqref{eq:community}.

\subsubsection{Calibration constraints reduce observed harm at a cost}
Table~\ref{tab:community-protection} reports the primary tolerance of one percentage point. All benefits in this comparison are relative to within-only communication at the same round count. A strict violation occurs if any community's test accuracy falls by more than the tolerance relative to that baseline. A secondary, prespecified material violation requires a loss exceeding two points. Point-constrained response selection reduces strict violations from 33.33\% to 17.78\% on digits and from 39.72\% to 4.44\% on XOR, while reducing mean benefits from 2.218 to 1.929 points and from 7.239 to 5.990 points, respectively. These percentages describe the evaluated blocks; they are not independent estimates of population-level risk.

\begin{table}[tb]
\centering\small
\input{tables/community_protection.tex}
\caption{Fixed-round community selection at $\varepsilon=1$ percentage point, averaged over 360 dependent blocks per task. Benefit is mean individual accuracy change relative to within-only communication, in percentage points (pp). Strict and material violation rates are percentages of blocks with any community losing more than 1 or 2 points, respectively. Direct calibration uses the same labels and candidate set. The guarded response's zero observed violations accompany frequent within-only fallback, not certified protection.}
\label{tab:community-protection}
\end{table}

Direct point-constrained calibration gives nearly identical mean benefits and slightly different violation rates. The constraint's value therefore cannot be attributed uniquely to response approximation. Guarded response selection has no observed strict violations, but chooses within-only communication in 97.5\% of digits blocks and 79.7\% of XOR blocks, leaving benefits of only 0.128 and 1.513 points. A fixed $\beta=0.5$ remains slightly better in mean XOR benefit than unconstrained response selection, while producing a higher fraction of blocks with a community violation. The observed tradeoff is between benefit and community protection, not uniform dominance by an adaptive selector.

\section{Discussion and limitations}

The results identify two kinds of information that aggregate descriptions omit. Spectrally matched operators can route signal and nuisance differently relative to a fixed task readout. A positive average gain can also conceal losses to particular communities or to a different group decision rule. Task-projected response addresses the first omission in the tested models; community-resolved evaluation exposes the second. The controlled counterexample, calibrated prediction tests and empirical benefit--harm tradeoff form the contribution. The underlying linearization, input--output geometry and constraint tools are standard.

The predictor requires white-box derivatives and target-condition labels, giving it more information than the scalar regressors. Direct calibration remains competitive, and timing advantages depend on the update form. Smooth, norm-constrained models with persistent private anchoring may favor local approximation. Only small update templates, three training seeds per study and one reused real-image split are tested. Directed corruption and operator rotation are stress tests; dependent grid counts and repeated partitions are not independent task replications. The observed individual gains do not establish superiority to global aggregation.

Community constraints reduce observed harm but leave a 17.78\% strict violation rate on digits at the primary tolerance. Guarded margins largely suppress inter-community communication, and zero observed violations are not certified protection. A separate streaming stress test found neither persistent error lock-in nor a general recovery advantage from reduced inter-community communication; local gating helped during XOR faults but performed similarly to exact feedback. Full results and a conditional contraction analysis appear in Appendix~\ref{app:dynamic}.

The community comparisons resolve individual, community and global outcomes, but do not establish a multiscale dynamical theory. Sizes of 8 and 16 cannot establish asymptotic scaling, especially when population size changes the evidence budget. A sign change in gain is not a phase transition. Further tests should connect response error to decision margins and evaluate new real tasks and distinct architectures against direct calibration at matched budgets. Formal protection would require valid finite-sample calibration with explicit treatment of surrogate error.

\section{Conclusion}
Matching complete operator spectra does not determine whether communication improves task accuracy. Retaining task-projected response gives accurate calibrated predictions in the small learned systems studied here, while community-resolved experiments show why positive mean gains need not benefit every community or readout. Calibration constraints reduce observed local harm at a cost in mean benefit, with direct calibration performing similarly. These results support a bounded account of communication gains that retains both task direction and the distribution of benefits; broad transfer, certified protection and practical superiority remain open.

\appendix
\section{Proof of spectral insufficiency}
\label{app:spectral-proof}
We prove Proposition~\ref{prop:spectral} under the assumptions stated in Section~\ref{sec:counterexample}.
\begin{proof}
The stacked Jacobian is $J_\theta=g(P\otimes W_\theta)$. With the orthogonal matrix $U_\theta=I_N\otimes R_\theta$,
\begin{equation}
 J_\theta=U_\theta J_0U_\theta^\top,\qquad
 S_T(\theta)=\sum_{k=0}^{T}J_\theta^k=U_\theta S_T(0)U_\theta^\top.
\end{equation}
Here, $h^T=S_T(\theta)x$. Orthogonal similarity preserves both eigenvalues and singular values, even when $P$ is not symmetric. For $\theta=0$, the first-channel readout is a nonnegative weighted sum of independent informative observations. Each row of $B=\sum_{k=0}^T(gaP)^k$ has at least two positive entries, so its signal-to-noise ratio is $\alpha\sum_jB_{ij}/\sqrt{\sum_jB_{ij}^2}>\alpha$ for $\alpha>0$.

Write $Q_T(\theta)$ for the mean accuracy after $T$ rounds at orientation $\theta$; $Q_0$ continues to denote accuracy without communication. For $\theta=\pi/4$, put $f_a=\sum_{k=0}^T(ga)^k$ and $f_b=\sum_{k=0}^T(gb)^k$. The readout mean conditioned on $Y=1$ is $\alpha(f_a+f_b)/2$, while shared nuisance alone contributes variance $\sigma^2r(f_a-f_b)^2/4$. Thus
\begin{equation}
 Q_T(\pi/4)\leq \Phi\!\left(\frac{\alpha(f_a+f_b)}{\sigma\sqrt r(f_a-f_b)}\right)<\Phi(\alpha)
 \quad\text{if}\quad
 \sigma\sqrt r>\frac{f_a+f_b}{f_a-f_b}.
 \label{eq:sufficient}
\end{equation}
Other noise contributions only decrease this upper bound. The condition is sufficient, not necessary.
\end{proof}
\section{Exact Gaussian reference}
Stack Eq.~\eqref{eq:data} as $x\mid Y=y\sim\mathcal N(ym,\Sigma)$, where $m=\alpha(\mathbf1_N\otimes e_1)$, $e_1=(1,0)^\top$, signal noise has covariance $I_N$ and nuisance covariance is $\sigma^2[(1-r)I_N+r\mathbf1\mathbf1^\top]$. The two channels have zero cross-covariance. If $C$ selects each agent's first coordinate and $A_T=CS_T$, the nondegenerate readout variance is $v_i=(A_T\Sigma A_T^\top)_{ii}$. Then
\begin{equation}
 Q_T=\frac1N\sum_i\Phi\left(\frac{(A_Tm)_i}{\sqrt{v_i}}\right).
\end{equation}
The calculation uses the known data generator and complete operator. Agreement with Monte Carlo checks the implementation and analytic construction, rather than prediction in an unknown system. At the primary setting, $f_a=3.3616$ and $f_b=1.2496$; the right side of Eq.~\eqref{eq:sufficient} is 2.18333, below $\sigma\sqrt r=2.82843$.

\section{Response ablations and proper-score decomposition}
The zero-message approximation uses $\widetilde h^{t+1}=h^0+\phi'(b)\odot(gP\widetilde h^tW)$, initialized at $h^0$. The readout ablation replaces $v$ in the displacement projection by $\|v\|w$, where $w$ is the calibrated orthogonal variance direction, while preserving the private logit. It tests this specific task projection and does not make $w$ label-independent.

For the prespecified secondary Brier endpoint, use $Y\in\{0,1\}$ and let $p_0$ and $\widetilde p_T$ denote the private and surrogate probabilities of $Y=1$, respectively, with $\Delta p=\widetilde p_T-p_0$. The surrogate Brier improvement satisfies the exact algebraic identity
\begin{equation}
 \E[(Y-p_0)^2-(Y-\widetilde p_T)^2]
 =2\E[(Y-p_0)\Delta p]-\E[(\Delta p)^2].
\end{equation}
The first term measures correction aligned with private prediction errors; the second penalizes the magnitude of the probability change. The identity is exact for surrogate probabilities. Their correspondence to the nonlinear model remains approximate, and the Brier endpoint is secondary to accuracy.

\section{Reproducibility and additional settings}
\label{app:repro}
Neural training and simulation use NumPy on CPUs; fitted predictive controls use scikit-learn. Finite differences checked gradients and one-step Jacobians for both activations. Additional checks verified zero gain at zero coupling and exact one-step proxy predictions. The maximum reference-radius matching error was $5.9\times10^{-15}$, and the Brier identity error was $1.1\times10^{-16}$. Retraining four selected task/activation models and their isolated counterparts reproduced bitwise-identical parameters in the recorded environment; regeneration also exactly matched 20 selected per-episode outcome arrays. These checks establish selective reproducibility, not independent replication of every condition.

The fitted regressors use 140 boosting iterations, learning rate 0.06, at most 15 leaves, minimum leaf size 20, L2 penalty 1 and no early stopping. Raw controls are $(g,T,N,\Iv,Q_{0,\mathrm{cal}},\theta,s)$, where $s$ is operator rescaling. One-probe controls use the one-step gain, $\Lambda$ and $\mu_T$. Training, basis selection, calibration and test generators have disjoint purpose-specific random seeds. Data reuse across topology and activation comparisons is intentional and recorded.

Nonlinear matching uses reference susceptibility at the private state. Let $\kappa_0$ and $\kappa_\theta$ denote its unit-coupling spectral radii for the base and rotated operators. The matched operator is $(\kappa_0/\kappa_\theta)W_\theta$, with no test gains used to determine the scale. This procedure preserves the reference radius, but not operator norm, the full Jacobian spectrum or the finite-time singular spectrum.

The preliminary neural experiment used 500 training minibatches, a single affine eight-channel private encoder followed by Softplus, and the same form of anchored messages. Its eight evidence conditions included independent and weak evidence, 50\%/100\% redundant views, 25\%/75\% misleading subpopulations, and 50\%/100\% common label shocks. Digits were divided into training/calibration/test source identities in proportions 50/25/25. A centralized learned reference saw all observed vectors but was not a mathematical accuracy upper bound. Details omitted from the main text, including the earlier spin-surface comparisons and baseline alignment supplements, remain in the research archive and are not used to claim successful universal transfer.

\section{Additional selection and intervention experiments}\label{app:followup}
\subsection{Selection in the original response study}
Within each of 144 task/activation/seed/size/graph/scenario blocks, we choose the candidate with highest predicted calibration gain, abstaining from communication if it does not exceed 0.005. Candidates include all operators, couplings and rounds; final test labels are not used for selection. Response-based selection gives mean test gain 14.631 percentage points, compared with 4.779 for the fixed base operator with $g=0.7$ and $T=3$, 10.646 for one-step selection and 14.636 for direct-calibration selection (Table~\ref{tab:selection}).

\begin{table}[tb]
\centering\small
\input{tables/selection.tex}
\caption{Communication selection over 144 condition blocks. Gain and regret are in percentage points (pp); regret is the difference between the test-oracle gain and the selected gain in the same block. Harm is the percentage of blocks with test gain below $-1$ percentage point. The test oracle selects the best candidate using test labels and serves only as a diagnostic upper reference. Different mean round counts preclude claims of equal-compute superiority.}
\label{tab:selection}
\end{table}

No response-selected block loses more than one percentage point of accuracy, although this observation provides no guarantee for new tasks. The response selector uses an average of 5.47 communication rounds, compared with 0.94 for one-step selection, and direct calibration remains slightly more accurate. Kernel timings exclude preprocessing and selection costs. The present results therefore establish neither equal-compute superiority nor an end-to-end efficiency advantage.

\subsection{Follow-up training, sampling and interventions}
The follow-up studies were designed sequentially after the response study. Twelve new models use XOR or the binary digits task, anchored or residual updates, and training seeds 411, 522 and 633. Encoders retain the 16-unit hidden layer and eight-channel state. Training uses 400 minibatches of 64, Adam with learning rate 0.003, the same message-norm limit of 0.75, $N=8$ Erd\H{o}s--R\'enyi graphs, coupling scheduled over $\{0,0.35,0.7,1\}$ and rounds alternating between $T=1$ and $T=3$. Digits use 898 training, 449 calibration and 450 test source images. An agent receives its disjoint flattened pixel subset plus a binary mask, represented by 128 inputs; the label is whether the digit is at least five. Pixel values are divided by 16. Total image evidence remains 64 pixels at both $N=8$ and $N=16$, whereas XOR keeps evidence per agent fixed.

The natural-intervention pilot evaluates rings and Erd\H{o}s--R\'enyi graphs, clean or shared-noise evidence, intact edges or independent edge deletion with probability 0.5, $g\in\{0.35,0.7,1\}$ and $T\in\{1,3,6\}$. Isolated nodes after deletion receive self-loops. It uses 32 or 128 calibration episodes and 512 XOR test episodes or all 450 test images. The 96 blocks contain 1,728 distinct candidate configurations; methods and label budgets reuse these candidates. Additional noise and timing settings appear below.

The community-bias study reuses these models at $g=0.7$, with four communities, $N\in\{8,16\}$, $T\in\{1,3,6\}$ and the four inter-community weights in Eq.~\eqref{eq:constraint}. One or two communities receive bias of strength $s\in\{0.5,1\}$, plus an unbiased control. XOR multiplies the second observation channel by $1-2s$. Digits add $4s$ times the difference between the training positive- and negative-class mean images to the affected observed pixels, without clipping. This directed stress test uses training labels to set the corruption direction, but does not query test labels. It is not representative random sensor noise. The study comprises 120 blocks and 1,440 configurations, using 128 calibration episodes and 1,024 XOR test episodes or 450 test images.

The constraint study uses three new community partitions, randomizes affected-community identities, and fixes $T=3$ or $T=6$ separately. Its 720 blocks contain 2,880 candidates. Tolerances $\varepsilon\in\{0,0.01,0.03\}$ are specified before the run, with 0.01 primary. Each block uses 128 calibration episodes and 1,024 new XOR test episodes or the same 450 digits test identities. Four scheduled message slots per agent per round are charged even at zero weight; active unique edges differ. Thus comparisons match scheduled slots within a fixed $(N,T)$, not physical communication across all settings. Later image studies reuse the test identities and cannot be treated as independent real-data replications.

Finally, a streaming stress test reuses the same 12 models without temporal retraining. Forty-eight model/size/partition blocks each contain four environment arms, six communication policies and 72 time steps. Targets stay fixed while private-observation noise is refreshed. A community fault operates during steps 24--47 between normal and recovery stages. Appendix~\ref{app:dynamic} specifies the local feedback controller, paired interventions and recovery endpoint. This study tests a further hypothesis about recovery; it is not additional evidence for a universal three-coordinate predictor.

\subsection{Natural-intervention results and computational controls}
With 128 calibration labels and one-step cases excluded from prediction evaluation, response RMSE on digits is 1.022 percentage points for anchored updates and 1.395 for residual updates. Selection includes the one-step candidates and abstention from communication. Response-selected gains over private decisions are 7.272 and 7.306 points; direct calibration gives 7.343 and 7.308. XOR response RMSE is 1.211 and 0.980 points, respectively. The candidate grid contains no gain below $-1$ point, so it cannot validate detection of harmful communication. Global averaging of private probabilities remains a strong reference: its accuracy is 70.685\% on digits and 96.354\% on XOR, compared with approximately 64.767\% and 90.546\% mean individual accuracy after response selection. These readouts have different information access. The extension supports calibrated prediction on images and natural edge interventions, not superiority to aggregation. Observed timing advantages occur only for the anchored update; the timing protocol and results are detailed below.

The natural-intervention models are trained independently of the earlier response-study models. Digits source identities are stratified using split seeds 731 and 732. Calibration subsets are drawn without replacement; training minibatches are drawn with replacement. Shared image noise adds an episode-specific Gaussian vector with standard deviation 0.6 before pixel masking. XOR shared noise uses the common-nuisance generator described in the main text. Static edge deletion retains each undirected edge independently with probability 0.5, followed by row normalization and self-loops for isolated nodes.

At 128 labels, the natural pilot's mean per-block ratios of response to direct runtime are 0.674 and 1.068 for anchored and residual digits models, and 0.689 and 1.059 for XOR. Timings are medians of three CPU calls, include encoding, candidate scoring and selection, and exclude training, data collection, graph construction and deployment. Fixed method order and millisecond durations limit interpretation. Selection gains average 12.715 points for response and 12.733 for direct calibration, compared with 9.849 for a fixed three-round setting; response selection uses 5.76 rounds on average. There is no equal-deployment-budget superiority claim.

The community constraint study randomizes timing order. Ratios of response to direct scoring time are approximately 0.72--0.73 for anchored models and 1.04--1.05 for residual models. Each scoring call supplies all tolerance and constraint variants, so these are not separate end-to-end deployment costs for individual policies. Source files retain the two additional tolerances, paired episode outcomes, seed summaries and strict as well as material violation indicators.

A post-run difference-in-differences diagnostic subtracts the clean inter-community effect from the biased effect on unaffected communities. Among 288 multi-round contrasts with one affected community, 202 show an additional loss exceeding one percentage point. This helps separate exposure to bias from the general reallocation of within- and between-community weights, but is a post hoc analysis of reused data, not independent confirmation of the mechanism.

\section{Streaming fault and recovery: results and protocol}\label{app:dynamic}
The streaming experiment does not support persistent error lock-in or a general recovery benefit from reducing inter-community communication. Under coherent faults, local response gating raises XOR fault-stage accuracy from 82.74\% under fixed $\beta=0.5$ to 84.46\%, but recovery is not faster. On digits, fault-stage improvement is only 0.17 points and recovery is slower. Exact one-step feedback performs almost identically. Fixed modular or intermittent policies generally reduce normal- and fault-stage accuracy; every coherent-fault block reaches its policy-specific mean-accuracy recovery criterion within the observation window.

Decorrelation with matched per-sensor fault counts produces negligible average changes, and resetting hidden states at recovery does not help. These interventions do not establish the proposed diversity--recovery mechanism. A post-run calculation explains part of the boundary: for identical private inputs and communication actions, state differences are bounded by 0.525 times the preceding difference for anchored models and 0.825 times that difference for residual models. This bound is not an accuracy guarantee and does not apply directly when feedback histories generate different actions. The protocol and all-policy summary below specify the scope of this calculation.

\subsection{Protocol and controls}

Two new partitions per model and size yield 48 blocks. Each stream has 24 normal, 24 fault and 24 recovery steps. XOR retains each target and latent private sign while drawing fresh unit-variance Gaussian noise; digits retain each source image and add independent pixel noise of standard deviation 0.15 before masking. The calibration stream has 128 targets, and test streams have 512 XOR targets or all 450 digits test images. Existing static-task weights are held fixed.

In the coherent arm, one randomly selected community receives a shared Bernoulli fault with probability 0.75 per target and time step. XOR reverses the second channel; digits apply the full-strength directed bias described in the main text. The clean arm supplies each policy's paired fault-free reference. The decorrelated arm independently permutes the coherent fault vector across targets for each affected sensor, preserving exactly its fault count at each step but potentially changing finite-sample pairing with class or difficulty. The reset arm matches the coherent arm through step 47 and, before the update at step 48, replaces hidden states with current private states; controller memory and pending feedback are retained. There is no additional clean-reset arm.

Policies use fixed $\beta=0.5$, fixed $\beta=0.125$, $\beta=0.5$ every fourth step and zero otherwise, or independently set each community's incoming inter-community weight to 0.5 with probability 0.25 and to zero otherwise. Two local policies compare within-only and $\beta=0.5$ shadow updates from the same previous calibration state. Each community receives its own sample feedback with probability 0.125, delayed three steps, and applies an exponentially weighted mean of the Brier improvement with rate 0.2. Gates start closed and open at an estimated improvement above 0.001. Test labels and fault identities do not enter the controller. Repeated feedback concerns the same 128 targets, so time-indexed feedback counts are not counts of newly labeled examples. Feedback is pooled within communities and gates are shared across the independent test trajectories.

The exact local policy computes both nonlinear shadow updates. The response policy instead linearizes the change in message weight at the within-only update:
\begin{equation}
 \widetilde h_{\rm cross}=h_{\rm in}
 +\tfrac12 g\,\phi'(a_{\rm in})\odot[(P_{\rm out}-P_{\rm in})hW],
\end{equation}
where $a_{\rm in}=b+gP_{\rm in}hW$ for anchored updates and $a_{\rm in}=b+P_{\rm in}hW$ for residual updates. This is a separately specified one-step message-direction approximation, not an unchanged application of the original finite-time response formula. All policies receive four scheduled message slots per node per step. Gating changes message use, not transmission count, and shadow evaluation adds computation.

Recovery lag is the first offset from step 48 at which three consecutive mean individual accuracies are no more than one percentage point below the same policy's clean reference; zero denotes immediate recovery and 24 denotes censoring. It does not require recovery of every community. Table~\ref{tab:dynamic} reports all policies under coherent faults. Across policies, decorrelation changes fault-stage accuracy by only $-0.049$ points on XOR and $-0.003$ on digits. Resetting lowers recovery-stage accuracy by 0.465 and 0.527 points, respectively. Without a clean-reset control, this does not isolate which components of retained state were useful.

\begin{table}[tb]
\centering\small
\input{tables/dynamic.tex}
\caption{Streaming coherent-fault results, averaged over 24 dependent blocks per task. Stage accuracies are percentages; lag is measured in steps relative to each policy's paired clean stream. Fixed 0.5 denotes the hierarchical graph with that inter-community weight, not an all-to-all graph. All blocks reach the stated mean-accuracy criterion within the recovery window.}
\label{tab:dynamic}
\end{table}

The post-run contraction calculation uses $\|W\|_2=0.75$ and $g=0.7$. For two trajectories with identical private inputs and identical row-stochastic communication actions, let $D_t$ be the maximum nodewise Euclidean state distance. Since Softplus and tanh are 1-Lipschitz, anchored updates satisfy $D_{t+1}\leq0.525D_t$, and residual updates satisfy $D_{t+1}\leq[(1-g)+g\|W\|_2]D_t=0.825D_t$. This bounds state differences conditional on common actions; accuracy also depends on decision margins, and feedback-controlled policies can select different actions after different histories.

\paragraph{Additional implementation checks.}
Follow-up checks cover gradients or directional derivatives, graph row sums, zero coupling, exact one-step response, image-identity separation and within-only isolation. Four representative blocks each in the natural, community-bias and constraint studies reproduced saved test arrays exactly. Dynamic checks verified matched fault counts, identical pre-reset trajectories and one complete regenerated trajectory batch. Flipping its test labels left gate actions unchanged. These are selective checks within the same research workflow, not an independent replication or a rerun of all model training.

\paragraph{AI assistance.}
AI tools assisted code development, mathematical exposition, analysis, literature search and manuscript drafting.

\bibliographystyle{plain}
\bibliography{references}
\end{document}

%% file: tables/prediction.tex
\begin{tabular}{lrrr}
\toprule
Predictor & $R^2$ & RMSE (pp) & Sign BA (\%) \\
\midrule
Task-aware response & 0.9988 & 0.454 & 99.95 \\
Zero-message approximation & 0.9981 & 0.568 & 99.95 \\
One-step gain & 0.8817 & 4.499 & 99.61 \\
Orthogonal readout & -0.6486 & 16.798 & 66.31 \\
Full direct calibration & 0.9989 & 0.426 & 99.95 \\
$\Lambda,\mu_T,\mathcal I_{\mathrm{vote}}$ & -0.1044 & 13.749 & 58.74 \\
Raw controls & 0.8389 & 5.252 & 99.56 \\
One-probe controls & 0.9178 & 3.750 & 95.38 \\
\bottomrule
\end{tabular}

%% file: tables/community_protection.tex
\begin{tabular}{llrrr}
\toprule
Task & Selector & Benefit (pp) & Strict (\%) & Material (\%) \\
\midrule
Digits & Fixed 0.5 & 1.963 & 46.39 & 28.06 \\
Digits & Response, unconstrained & 2.218 & 33.33 & 20.56 \\
Digits & Response, point constraint & 1.929 & 17.78 & 5.28 \\
Digits & Direct, point constraint & 1.925 & 17.22 & 4.44 \\
Digits & Response, guarded & 0.128 & 0.00 & 0.00 \\
\midrule
XOR & Fixed 0.5 & 7.317 & 46.11 & 37.78 \\
XOR & Response, unconstrained & 7.239 & 39.72 & 32.50 \\
XOR & Response, point constraint & 5.990 & 4.44 & 1.39 \\
XOR & Direct, point constraint & 5.974 & 4.72 & 1.67 \\
XOR & Response, guarded & 1.513 & 0.00 & 0.00 \\
\bottomrule
\end{tabular}

%% file: tables/selection.tex
\begin{tabular}{lrrrr}
\toprule
Selector & Gain (pp) & Regret (pp) & Harm (\%) & Rounds \\
\midrule
No communication & 0.000 & 14.670 & 0.0 & 0.00 \\
Fixed setting & 4.779 & 9.891 & 50.0 & 3.00 \\
One-step selection & 10.646 & 4.023 & 0.0 & 0.94 \\
Response selection & 14.631 & 0.039 & 0.0 & 5.47 \\
Direct calibration & 14.636 & 0.033 & 0.0 & 5.49 \\
Test oracle & 14.670 & 0.000 & 0.0 & 5.62 \\
\bottomrule
\end{tabular}

%% file: tables/dynamic.tex
\begin{tabular}{llrrrr}
\toprule
Task & Policy & Normal (\%) & Fault (\%) & Recovery (\%) & Lag \\
\midrule
XOR & Fixed 0.5 & 93.85 & 82.74 & 94.30 & 3.00 \\
XOR & Fixed 0.125 & 92.69 & 79.26 & 92.98 & 3.71 \\
XOR & Intermittent & 92.55 & 79.27 & 92.92 & 3.54 \\
XOR & Random gate & 92.51 & 78.98 & 92.91 & 3.12 \\
XOR & Local response & 93.43 & 84.46 & 94.31 & 3.21 \\
XOR & Local exact & 93.43 & 84.46 & 94.31 & 3.25 \\
\midrule
Digits & Fixed 0.5 & 68.14 & 65.81 & 68.72 & 1.46 \\
Digits & Fixed 0.125 & 66.34 & 64.09 & 66.75 & 1.46 \\
Digits & Intermittent & 66.42 & 64.15 & 66.87 & 1.42 \\
Digits & Random gate & 66.30 & 64.14 & 66.64 & 1.25 \\
Digits & Local response & 68.03 & 65.98 & 68.93 & 2.88 \\
Digits & Local exact & 68.04 & 65.98 & 68.93 & 2.96 \\
\bottomrule
\end{tabular}